%% file: main.tex
\pdfoutput=1
\documentclass[11pt]{article}
\usepackage{header}
\usepackage{macro}

\title{Separations between Oblivious and Adaptive Adversaries\\for Natural Dynamic Graph Problems}
\author{%
    Aaron Bernstein\thanks{New York University, \texttt{bernstei@gmail.com}. Supported by Sloan Fellowship, Google Research Fellowship, NSF Grant 1942010, and Charles S.\ Baylis endowment at NYU.} \and
    Sayan Bhattacharya\thanks{University of Warwick, \texttt{S.Bhattacharya@warwick.ac.uk}.} \and
    Nick Fischer\thanks{Max Planck Institute for Informatics, \texttt{nfischer@mpi-inf.mpg.de}. Parts of this work were done while the author was affiliated with INSAIT, Sofia University ``St.\ Kliment Ohridski''. This work is partially funded by the Ministry of Education and Science of Bulgaria (support for INSAIT, part of the Bulgarian National Roadmap for Research Infrastructure).} \and
    Peter Kiss\thanks{University of Vienna, \texttt{peter.kiss@univie.ac.at}. This research was funded in whole or in part by the Austrian Science Fund (FWF) 10.55776/ESP6088024}\and
    Thatchaphol Saranurak\thanks{University of Michigan, \texttt{thsa@umich.edu}. Supported by NSF Grant CCF-2238138.}}
\date{}

\begin{document}

\maketitle
\pagenumbering{gobble}

\input{abstract}

\newpage
\tableofcontents

\newpage
\pagenumbering{arabic}

\input{sections/intro}
\input{sections/prelim}
\input{sections/reduction}
\input{sections/triangle}
\input{sections/clique-algo}
\input{sections/open}

\bibliographystyle{alphaurl}
\bibliography{refs}

\end{document}

%% file: abstract.tex
\begin{abstract}
\noindent
We establish the first update-time separation between dynamic algorithms against \emph{oblivious} adversaries and those against \emph{adaptive} adversaries in natural dynamic graph problems, based on popular fine-grained complexity hypotheses.

Specifically, under the combinatorial BMM hypothesis, we show that every combinatorial algorithm against an adaptive adversary for the incremental maximal independent set problem requires $n^{1-o(1)}$ amortized update time. Furthermore, assuming either the 3SUM or APSP hypotheses, every algorithm for the decremental maximal clique problem needs $\Delta/n^{o(1)}$ amortized update time when the initial maximum degree is $\Delta \le \sqrt{n}$. 
These lower bounds are matched by existing algorithms against adaptive adversaries. In contrast, both problems admit algorithms against oblivious adversaries that achieve $\operatorname{polylog}(n)$ amortized update time \cite{BehnezhadDHSS19,ChechikZ19}. Therefore, our separations are exponential.

Previously known separations for dynamic algorithms were either engineered for contrived problems and relied on strong cryptographic assumptions \cite{BeimelKMNSS22}, or worked for problems whose inputs are not explicitly given but are accessed through oracle calls \cite{BateniEFHJMW23}.

As a byproduct, we also provide a separation between incremental and decremental algorithms for the triangle detection problem: we show a decremental algorithm with $\tilde{O}(n^{\omega})$ total update time, while every incremental algorithm requires $n^{3-o(1)}$ total update time, assuming the OMv hypothesis. To our knowledge this is the first separation of this kind.
\end{abstract}

%% file: sections/intro.tex
\section{Introduction}
In the setting of dynamic algorithms, an \emph{adversary} generates a sequence of updates to the input, and a dynamic algorithm aims to maintain useful information about the input throughout the updates with minimal update time. The adversary is \emph{oblivious} if its update sequence is fixed from the beginning and, in particular, is independent of the algorithm's output. Otherwise,  the adversary is \emph{adaptive} if the next update can be adaptively chosen based on the algorithm's previous outputs.

An algorithm is \emph{robust} or works against an adaptive adversary if its performance is guaranteed even against an adaptive adversary. Otherwise, if it only works against an oblivious adversary, it is called \emph{non-robust}. Robustness against adaptive adversaries is needed in a wide range of real-world applications and is well-studied across many areas of computer sciences, such as streaming algorithms~\cite{Ben-EliezerJWY22,HassidimKMMS22,AssadiCGS23}, online algorithms~\cite{Ben-DavidBKTW94}, and adaptive data analysis~\cite{DworkFHPRR15,HardtU14}. Robustness is especially critical for dynamic algorithms as it enables their use as subroutines in static algorithms, which is a key application of dynamic algorithms.

\paragraph{Success in Making Algorithms Robust.}
For many dynamic graph problems, the first efficient algorithms were non-robust, with robust algorithms lagging far behind. But while many gaps remain, in the past decade, we have witnessed many successful lines of work on developing robust dynamic algorithms that match the near-optimal guarantees of the non-robust ones~\cite{RodittyZ08,KapronKM13,BaswanaKS12,HenzingerKN18,AbrahamDKKP16,Chechik18,BernsteinGW19}. This includes robust algorithms for spanning forests~\cite{NanongkaiS17,Wulff-Nilsen17,NanongkaiSW17,ChuzhoyGLNPS20,GoranciRST21}, approximate shortest paths~\cite{BernsteinC16,BernsteinC17,Chuzhoy21,ChuzhoyS21,BernsteinGS21,ChuzhoyZ23,HaeuplerLS24}, graph sparsifiers~\cite{ChenGHPS20,BernsteinBGNSSS22,BhattacharyaSS22,KyngMG24}, and reachability~\cite{BernsteinGS20,ChenKLMG24,BrandCKLMGS24}. These data structures and their techniques have been the core powerhouse behind breakthroughs in static algorithms, including almost-linear time algorithms for min-cost flow~\cite{ChenKLPGS25,BrandCPKLGSS23} and Gomory-Hu trees~\cite{AbboudKLPST22,AbboudLPS23}.

Recently, another significant barrier was overcome for symmetry-breaking problems, which include maximal independent set, maximal matching, and $(\Delta+1)$-vertex coloring. These problems admit near-optimal non-robust algorithms with polylogarithmic update time~\cite{BaswanaGS18,BhattacharyaCHN18,ChechikZ19,BehnezhadDHSS19}, but no robust algorithms with $o(n)$ update time were previously known. This past year, however,~\cite{BehnezhadRW25} and~\cite{BernsteinBKS25} introduced new robust dynamic algorithms with truly sublinear $n^{1-\Omega(1)}$ update times for $(\Delta+1)$-vertex coloring and maximal matching, respectively. Given these positive evidences over the years, can we hope to always almost match the performance of non-robust algorithms with robust ones?

\paragraph{Previous Separations.}
Strictly speaking, the answer is no. Consider this artificial problem. Initially, we are given an array $A=[A_1, \dots, A_n]$ of $n$ numbers. For each update, the adversary chooses an index $i \in [n]$ and marks it as deleted. The goal is to maintain a tuple $(i, A_i)$ where $i$ is not yet deleted, if it exists. Observe that an adaptive adversary can force every algorithm to update its answers $n$ times by deleting $i$ whenever $(i, A_i)$ is outputted. But, assuming an oblivious adversary, if we maintain $(i, A_i)$ where $i$ is a random undeleted index, then our maintained index $i$ will be deleted at most $O(\log n)$ times in total with high probability. Thus, only $O(\log n)$ entries need to be queried. Therefore, if the input array is provided as an oracle, it is quite simple to achieve an exponential separation between oblivious and adaptive adversaries.

The above separation, however, has two drawbacks:
\begin{enumerate}
\item The problem is artificial, and
\item The input must be prepared as an oracle. This model does not capture most standard dynamic problems.
\end{enumerate}
Regarding the second drawback, note that the above $O(\log n)$-time algorithm cannot even read the entire input. In contrast, in most dynamic graph, string, and geometric problems, the input and updates are given explicitly, requiring every algorithm to spend at least linear time to read the input.

Unfortunately, all known separations for dynamic algorithms still suffer from one of the above drawbacks. Specifically, the separations by Beimal et al.~\cite{BeimelKMNSS22} were engineered for artificial problems, one of them similar to the example above, and also assumed strong cryptographic assumptions. Another separation by Bateni et al.~\cite{BateniEFHJMW23} was for dynamic $k$-clustering problems which assume an oracle for the input and are hence not using the standard measure of update time. Furthermore, their lower bound only applies to deterministic algorithms and does not strictly separate robust from non-robust algorithms.\footnote{The adversary in the lower bound of~\cite{BateniEFHJMW23} is \emph{metric-adaptive} and, so it cannot capture most randomized $k$-clustering algorithms even in the static setting (see, e.g.,~\cite{MettuP04}) which assume a weaker \emph{point-adaptive} adversary. The robust randomized algorithm by~\cite{BhattacharyaCF25} beats the lower bound of~\cite{BateniEFHJMW23} because of this.} Until now, there is no oblivious-vs-adaptive separation for \emph{natural} dynamic problems with explicit input.

\subsection{Our Results: Natural Oblivious-vs-Adaptive Separations}
We establish the first update-time separations between non-robust and robust dynamic algorithms for natural dynamic problems with explicit input, assuming popular fine-grained complexity hypotheses, detailed in \Cref{subsec:hypo}. Recall that incremental (decremental) graphs are graphs undergoing only edge insertions (deletions).

\paragraph{Incremental MIS.}
Our first separation is for the dynamic \emph{maximal independent set} (MIS) problem~\cite{AssadiOSS18,AssadiOSS19,BehnezhadDHSS19,ChechikZ19,GuptaK21}. In this problem, given an $n$-vertex graph undergoing edge updates, we must maintain an MIS. This problem admits fully dynamic non-robust algorithms with $\text{polylog}(n)$ update time, handling both edge insertions and deletions~\cite{BehnezhadDHSS19,ChechikZ19}. We show that, even in the incremental setting, every robust algorithm requires large update time.

\begin{theorem} \label{thm:hardness MIS}
Assuming the Boolean-Matrix-Multiplication (BMM) hypothesis, every algorithm against an adaptive adversary for maintaining an MIS on an incremental graph with $n$ vertices requires $\Omega(n^{\omega-o(1)})$ total update time, where $\omega\le2.372$ is the matrix multiplication exponent. That is, the amortized update time is at least $\Omega(n^{\omega-2-o(1)})$.

For combinatorial\footnote{\emph{Combinatorial} algorithms informally refer to algorithms that do not use fast matrix multiplication.} algorithms, the lower bound improves to $\Omega(n^{1-o(1)})$ amortized update time.
\end{theorem}

Assuming $\omega>2$, this gives an exponential separation since every robust algorithm requires polynomial update time, while non-robust algorithms only take polylogarithmic update time. Our lower bound is also \emph{tight} with known combinatorial robust algorithms~\cite{AssadiOSS18,GuptaK21} that achieve $O(n)$ amortized update time in dense graphs.

\paragraph{Decremental Maximal Clique.}
Our second separation is exponential even when $\omega=2$. Below, we present a tight lower bound for robust algorithms for the decremental \emph{maximal clique} problem.

\begin{theorem}\label{thm:hardness clique}
Assuming the 3SUM or All-Pairs Shortest Paths (APSP) hypothesis, every algorithm against an adaptive adversary for maintaining a maximal clique in each connected component of a decremental graph with $n$ vertices and initial maximum degree $\Delta\le\sqrt{n}$ requires at least $\Omega((n \Delta^2)^{1-o(1)})$ total update time. That is, the amortized update time is at least $\Omega(\Delta/n^{o(1)})$.
\end{theorem}

The same lower bound also applies to maintaining a \emph{single} maximal clique if vertex deletions are also allowed. The requirement on every connected component is simply to prevent algorithms from always outputting a trivial set $\{v\}$ once a singleton connected component $\{v\}$ occurs.

To complete the separation, we also show a non-robust algorithm with near-optimal update time.

\begin{theorem}\label{thm:alg clique}
There is an algorithm against an oblivious adversary that, with high probability, maintains a maximal clique in every connected component of a decremental graph with $n$ vertices in $\polylog(n)$ amortized update time.
\end{theorem}

Note that, although a maximal clique is an MIS on the complement graph, \Cref{thm:alg clique} does not immediately follow from dynamic MIS algorithms for two reasons. The main reason is that the complement graph can be significantly denser than the original graph, which could lead to an $O(n^2)$ \emph{initialization time} in a sparse graph. A second reason is that Theorem \ref{thm:alg clique} needs to maintain a maximal clique in every connected component.  We show that the theorem can nonetheless be obtained via a simple reduction. 

\Cref{thm:hardness clique,thm:alg clique} together provide an exponential separation even when $\omega=2$. Moreover, our lower bound is tight with a robust decremental maximal clique algorithm that has $O(\Delta)$ amortized update time, as shown in \Cref{sec:clique}.

Finally, we emphasize that while establishing a separation between adaptive and oblivious adversaries for natural problems has been a longstanding open question in the dynamic algorithms community, our main results are surprisingly simple and can be taught in undergraduate lectures.
The hardness of maximal clique (\Cref{thm:hardness clique}) is via a reduction from triangle detection that is even ``straightforward in hindsight''. Via a simple extension, we then obtain the hardness of MIS (\Cref{thm:hardness MIS}). Given the significance of the question, we view this simplicity as a virtue.

\subsection{Byproduct: First Incremental-vs-Decremental Separation}
\Cref{thm:hardness MIS} leaves open the question of whether there is a separation for dynamic MIS when $\omega=2$. Our further investigation into this question surprisingly leads us to the first separation between the incremental and decremental settings. Below, we discuss this interesting byproduct and then explain its implication for our original motivation.

\paragraph{Incremental-vs-Decremental Separation.}
Recall that a triangle is a clique of size 3. We show that triangle detection in incremental graphs is strictly harder than in decremental graphs.

\begin{restatable}[Incremental triangle detection is hard]{theorem}{thminctriangle}\label{thm:inc triangle}
Assuming the OMv hypothesis, every algorithm that can report YES once there exists a triangle in an incremental graph requires $n^{3-o(1)}$ total update time.
\end{restatable}

Note that the above lower bound is for the decision version of the triangle detection problem and, consequently, holds even when an adversary is oblivious.\footnote{For decision problems the algorithm's answer is uniquely determined, so an adaptive adversary is not stronger than an oblivious one.}

\begin{restatable}[Decremental triangle detection is easy]{theorem}{thmdectriangle}\label{thm:dec triangle}
There is a randomized algorithm against an adaptive adversary that, with high probability, maintains a triangle in a decremental graph and reports NO once no triangle exists anymore in $\tilde{O}(n^{\omega})$ total update time.
\end{restatable}

The above algorithm is near-optimal because its total update time nearly matches the state-of-the-art \emph{static} triangle detection algorithm with $O(n^{\omega})$ running time. It is also robust against an adaptive adversary and works for the search version, i.e., a triangle is maintained explicitly.

\Cref{thm:inc triangle,thm:dec triangle} together establish an 
update-time separation between the incremental and decremental versions of the a problem; as far as we know, this is the first
separation of this kind.

\paragraph{A Barrier for Our Incremental MIS Lower Bound.}
\Cref{thm:dec triangle} shows a barrier to improving our lower bound for dynamic MIS. This is because the proof of \Cref{thm:hardness MIS} is based on showing the hardness of the weaker \emph{decremental 3-maximal clique} problem, where the goal is to maintain a clique $S$ that is maximal \emph{or} has size $|S| \ge 3$ in a decremental graph. That is, the maximality property is not required when $|S| \ge 3$. However, via \Cref{thm:dec triangle}, we can solve the decremental 3-maximal clique problem in $\tilde{O}(n^{\omega})$ total update time.

\begin{corollary} \label{cor:3-mis}
There is a randomized algorithm against an adaptive adversary for solving the decremental 3-maximal clique problem with high probability in $\tilde{O}(n^{\omega})$ total update time.
\end{corollary}
\begin{proof}
Run the decremental triangle detection algorithm from \Cref{thm:dec triangle} on $G$. As long as there exists a triangle $T$ in $G$, report $T$. Then, as long as there exists an edge $e$, report $e$. Finally, once $G$ is empty, report any vertex. This is correct because once no triangle remains in $G$, every edge forms a \emph{maximal} clique. When $G$ becomes empty, any vertex is a maximal clique.
\end{proof}

This indicates that our lower bound approach in \Cref{thm:hardness MIS} is tight. To achieve a $\Omega(n^{2.1})$ lower bound when $\omega = 2$, one might consider showing the hardness of the dynamic 4-maximal clique problem, where maximality is not required when $|S| \ge 4$, for example. We leave this as an intriguing open problem.

\paragraph{Organization.} 
After the preliminaries in \Cref{sec:prelim}, we show the hardness of incremental MIS and decremental maximal clique in \Cref{sec:reduc}, proving \Cref{thm:hardness MIS,thm:hardness clique}. The incremental-decremental separation for triangle detection is presented in \Cref{sec:triangle}, proving \Cref{thm:inc triangle,thm:dec triangle}. Next, we give a simple reduction that implies a near-optimal non-robust decremental maximal clique algorithm in \Cref{sec:clique}, proving \Cref{thm:alg clique}. Finally, in \Cref{sec:open}, we conclude a list of interesting open problems generated by our results. We also survey prior separation results related to ours.

%% file: sections/prelim.tex
\section{Preliminaries}
\label{sec:prelim}
Throughout all graphs $G = (V, E)$ are undirected and unweighted. Let $G^c$ denote the \emph{complement} graph of $G$, i.e., the graph containing an edge $\{x, y\}$ if and only if $\{x, y\} \not\in E$. An \emph{independent set} is a set of vertices $I \subseteq V$ such that all pairs of distinct vertices $x, y \in I$ do not form an edge, $\{x, y\} \not\in E$. A \emph{maximal} independent set (MIS) is an independent set $I$ such that $I \cup \{x\}$ is not an independent set for all vertices $x \in V \setminus I$. In contrast, a \emph{clique} $K \subseteq V$ is a set of vertices such that all pairs of distinct vertices $x, y \in K$ form an edge, $\{x, y\} \in E$, and we define \emph{maximal} cliques analogously.

\begin{observation} \label{obs:mis-complement}
A set $I \subseteq V$ is a (maximal) independent set in $G$ if and only if $I$ is a (maximal) clique in the complement graph $G^c$.
\end{observation}

\paragraph{Fine-Grained Complexity Hypotheses.} \label{subsec:hypo}
In the following we summarize the fine-grained hardness assumptions used throughout. These are among the most popular assumptions in fine-grained complexity; see e.g.~\cite{Vassilevska19}. All hypotheses implicitly allow randomized algorithms.

\begin{hypothesis}[Combinatorial BMM]\label{conj:comb}
Every combinatorial algorithm for Boolean Matrix Multiplication requires at least $\Omega(n^{3-o(1)})$ time. 
\end{hypothesis}

\begin{hypothesis}[BMM]\label{conj:bmm}
Every algorithm for Boolean Matrix Multiplication requires at least $\Omega(n^{\omega-o(1)})$ time, where $\omega$ is the matrix multiplication exponent.  
\end{hypothesis}

\begin{hypothesis}[3SUM]\label{conj:3sum}
Every algorithm that decides if a size-$n$ set $A \subseteq \{0, \dots, n^C\}$ (for some constant $C$) contains elements $a, b, c \in A$ with $a + b = c$ requires at least $\Omega(n^{2-o(1)})$ time. 
\end{hypothesis}

\begin{hypothesis}[APSP]\label{conj:apsp}
Every algorithm that computes the pairwise distances in an $n$-vertex directed graph with edge weights $\{0, \dots, n^C\}$ (for some constant $C$) requires $\Omega(n^{3-o(1)})$ time.
\end{hypothesis}

\begin{hypothesis}[OMv]\label{conj:omv}
Every algorithm that preprocesses a matrix $M \in \{0, 1\}^{n \times n}$ and then one-by-one answers $n$ Boolean matrix-vector queries (given $v \in \{0, 1\}^n$, report \smash{$\bigvee_{j \in [n]} (M_{ij} \land v_j)$} for all $i \in [n]$) requires total time $\Omega(n^{3-o(1)})$.
\end{hypothesis}

%% file: sections/reduction.tex
\section{Hardness Reductions} \label{sec:reduc}

\paragraph{Decremental Maximal Clique.}
First, we show a simple reduction from triangle detection to decremental maximal clique. In fact, our reduction is based on the following \emph{all-edges} variant which is well-known to capture the BMM problem in dense graphs (see \Cref{obs:all-edges-triangle-bmm}) and is hard even in sparse graphs under both the 3SUM~\cite{Patrascu10} and APSP~\cite{WilliamsX20,ChanWX22} hypotheses (see \Cref{thm:all-edges-triangle-hardness}).

\begin{definition}[All-Edges Triangle Detection]
Given a tripartite graph $G = (X, Y, Z, E)$, decide for each edge $(x, z) \in (X \times Z) \cap E$ if it is involved in a triangle $(x, y, z) \in X \times Y \times Z$.
\end{definition}

\begin{observation}[E.g.~\cite{WilliamsW18}] \label{obs:all-edges-triangle-bmm}
For any constant $\epsilon > 0$, there is no $O(n^{\omega-\epsilon})$-time algorithm for the All-Edges Triangle Detection problem unless the BMM Hypothesis fails. Similarly, there is no $O(n^{3-\epsilon})$-time \emph{combinatorial} algorithm unless the \emph{Combinatorial} BMM Hypothesis fails.
\end{observation}

\begin{theorem}[\cite{Patrascu10}, {{\cite[Corollary~1.4]{WilliamsX20}}}, \cite{ChanWX22}]\label{thm:all-edges-triangle-hardness}
For any constants $0 < \delta \leq \frac{1}{2}$ and $\epsilon > 0$, there is no $O((n \Delta^2)^{1-\epsilon})$-time algorithm for the All-Edges Triangle Detection problem in graphs with maximum degree $\Delta = O(n^\delta)$, unless both the 3SUM and APSP hypotheses fail.
\end{theorem}

As mentioned in the introduction (Theorem \ref{thm:hardness clique}), our reduction relies on the mild assumption that the algorithm maintains a maximal clique for each connected component. Note that our algorithm from \cref{thm:alg clique} can do this.

\begin{lemma}
Suppose there is a decremental algorithm ${\cal A}$ for maintaining a maximal clique in each connected component of a graph with $n$ vertices and initial maximum degree $\Delta$ in $T(n, \Delta)$ total update time against an adaptive adversary. Then the All-Edges Triangle Detection problem on $n$-vertex graphs with maximum degree $\Delta$ can be solved in time $O(T(n, \Delta))$.
\end{lemma}
\begin{proof}
Let $G = (X, Y, Z, E)$ denote the given All-Edges Triangle Detection instance, and initialize $G' \gets G$. We run ${\cal A}$ on $G'$ and will adaptively generate a sequence of edge deletions to $G'$ as follows. Take an arbitrary non-singleton connected component $C$, and let $K_{C}$ denote the maximal clique in $C$ maintained by ${\cal A}$. Note that $|K_C| \geq 2$ (since $C$ is not a singleton) and $|K_C| \leq 3$ (since $G$ is tripartite). This leaves two cases:
\begin{enumerate}
    \item If $|K_{C}| = 3$, then $K_{C}$ is a triangle of the form $(x, y, z) \in X \times Y \times Z$. We report YES for the edge $(x, z)$, remove the edge $(x, z)$ from the graph $G'$ and continue.
    \item If $|K_{C}| = 2$, then $K_C$ is an edge of the form $(x, y) \in X \times Y$, $(y, z) \in Y \times Z$ or $(x, z) \in X \times Z$. In either case we remove the edge and continue.
\end{enumerate}
When the graph $G'$ becomes empty we terminate (reporting NO for all edges $(x, z)$ for which we have not previously reported YES).

It is clear that we only report YES for edges $(x, z)$ that are actually involved in a triangle. Conversely, focus on any edge $(x, z)$ that is involved in a triangle. We argue that in each step the algorithm either detects a triangle $(x, y, z)$ (in which case it reports YES), or it deletes an edge that leaves all triangles involving $(x, z)$ intact. Indeed, if in case 1 the algorithm detects some other triangle~$(x', y', z')$ with $(x', z') \neq (x, z)$ then we remove an irrelevant edge $(x', z')$. And if in case 2 the algorithm removes some edge $(u,v)$, then by the maximality of $K_C = \{u,v\}$ that edge is not contained in any triangle. Since the graph is eventually empty, it follows that we have reported YES for $(x, z)$ at some point.

Finally, observe that the additional time we spent outside running ${\cal A}$ is $O(n \Delta)$. Therefore, the total running time is $O(n \Delta + T(n, \Delta)) = O(T(n, \Delta))$. (Here, we use the trivial bound $T(n, \Delta) = \Omega(n \Delta)$ as in the worst case $\mathcal A$ processes $\Omega(n \Delta)$ updates.)
\end{proof}

Combined with \Cref{thm:all-edges-triangle-hardness}, the above reduction immediately implies \Cref{thm:hardness clique}.

\begin{remark}
Suppose instead there is a \emph{fully dynamic} algorithm ${\cal A}'$ that simply maintains a maximal clique in $G$, without the additional requirement of doing so for every connected component. We can use ${\cal A}'$ to solve triangle detection using $O(n^{2})$ edge updates by the following similar reduction. Let $K$ be a maximal clique maintained by ${\cal A}'$. If $|K|\ge3$, then we find a triangle and terminate. If~\makebox{$|K|=2$}, say $K=\{u,v\}$, then $\{u,v\}$ is not in any triangle and we delete the edge $\{u,v\}$. Else, say~\makebox{$K=\{v\}$}, then $v$ is not in any triangle. In this case, we insert edges from $v$ to all other vertices $w$ so that any maximal clique must contain $v$. From now, we\emph{ mark} $v$ and ignore $v$. More specifically, the algorithm maintains the set of marked vertices $M$. If $|K\setminus M|\ge3$, we find a triangle and terminate. If $|K\setminus M|=2$, we delete the edge in $K\setminus M$. Else, $K\setminus M=\{v\}$, then we insert edges from $v$ to all other vertices, mark $v$. We continue until we find a triangle, or all vertices are marked, in which case no triangle is left.
\end{remark}

\paragraph{Incremental Maximal Independent Set.}
An independent set is a clique in the complement graph. Hence, we can adjust the above reduction to become a reduction for incremental MIS by working on the complement graph $G'$ instead. One technical issue is that, in our incremental graph, once there exists a vertex $v$ that is adjacent (in $G$) to every other vertex, then the algorithm can keep reporting $\{v\}$ as its MIS, which is not informative for the reduction. We adjust the reduction to prevent this scenario by making two copies of the complement graph. The complete reduction is described below.

\begin{lemma}
Suppose there is an incremental algorithm ${\cal A}$ for maintaining an MIS on a graph with $n$ vertices in $T(n)$ total update time against an adaptive adversary. Then, we can check if an $n$-vertex graph contains a triangle in $O(T(2n))$ time.
\end{lemma}
\begin{proof}
Initialize the graph $G'$ that contains two copies of the complement graph of $G$, denoted by~$G'_{1}$ and~$G'_{2}$. There is initially no edge between $V(G_{1}')$ and $V(G_{2}')$ in $G'$. For each vertex $u$ in $G$, let $u_{1},u_{2}$ denote the corresponding vertices in $G'_{1},G'_{2}$, respectively. 

We run ${\cal A}$ on $G'$ and will adaptively generate a sequence of edge insertions to $G'$ such that $G'_{1}$ and $G'_{2}$ are always identical. We will \emph{mark} vertices of $G$. Once $u$ is marked, then both $u_{1}$ and $u_{2}$ are marked. Initially, there is no marked vertex. The invariant is as follows:
\begin{enumerate}[label=(\Roman*)]
\item \label{enu:inv1}The edges between $V(G_{1}')$ and $V(G_{2}')$ form precisely the bi-clique between marked vertices in $V(G_{1}')$ and $V(G_{2}')$.
\item \label{enu:inv2}Once $u$ is marked, then $u_{i}$ is adjacent to every other vertex in $G_{i}'$ for $i\in\{1, 2\}$.
\item \label{enu:inv3}For $i\in\{1,2\}$, if $\{u_{i},v_{i}\}$ is inserted to $G'_{i}$ , then $\{u,v\}$ is an edge in $G$ but $\{u,v\}$ is not in any triangle of $G$.
\end{enumerate}
So, once $G'_{i}$ becomes a complete graph, then we can conclude that $G$ has no triangle. Also, once $u$ is marked, then $u$ is not in any triangle of $G$.

Let $K$ denote the MIS in $G'$ maintained by ${\cal A}$. By Invariant \ref{enu:inv1} and \ref{enu:inv2}, all marked vertices are adjacent to each other. So, there is at most one marked vertex in $K$. For $i\in\{1,2\}$, let $K_{i}=K\cap V(G_{i})$. Suppose without loss of generality that $K_{1}$ contains no marked vertices. There are three cases:

\begin{enumerate}
    \setlength\parskip{0pt plus 1pt minus 1pt}
    \setlength\parindent{1.6em}
    \item If $|K_{1}|=1$, say $K_{1}=\{v_{1}\}$, then we mark $v$ and add all edges from $v_1$ to all marked vertices in~$G'_2$ and from $v_2$ to all marked vertices in $G'_1$, in order to satisfy Invariant~\ref{enu:inv1}. Invariant~\ref{enu:inv2} still holds because we claim that every vertex $w_{1}\in V(G'_{1})$ is adjacent to $v_{1}$. Indeed, if $w_{1}$ is marked, then $w_{1}$ is adjacent to $v_{1}$ by Invariant~\ref{enu:inv2}. If $w_{1}$ is unmarked, $w_{1}$ must be adjacent to~$v_{1}$ too, otherwise $K\cup\{w_{1}\}$ is independent in $G'$ (here we are using the fact that an unmarked vertex in $G'_1$ has no edges to $G'_2$), which contradicts the maximality of $K$.
    
    \item If $|K_{1}|=2$, say $K_{1}=\{u_{1},v_{1}\}$, then we insert the edges $\{u_{i},v_{i}\}$ into $G'_{i}$ for $i\in\{1,2\}$. We will prove that Invariant~\ref{enu:inv3} still holds. Indeed, $\{u,v\}$ is an edge in $G$ because $\{u_{1},v_{1}\}$ was a non-edge in $G'_{1}$ before insertion. It remains to prove that $\{u,v\}$ is not in any triangle in $G$.

    Suppose for contradiction that $\{u,v,w\}$ is a triangle in $G$ for some vertex $w$. First, observe that $w_{1}$ is either adjacent to $u_{1}$ or $v_{1}$ in the current $G'_1$. Indeed, this holds when $w_{1}$ is marked by Invariant~\ref{enu:inv2}; if $w_{1}$ is unmarked and $w_{1}$ is not adjacent to both $u_{1}$ and $v_{1}$, then $K\cup\{w_{1}\}$ is independent, which contradicts the maximality of $K$. Suppose without loss of generality that~$w_{1}$ is adjacent to $u_{1}$ in $G'_1$. There are two subcases. First, if $\{u_{1},w_{1}\}$ was initially in~$G'_{1}$ before any insertion, then $\{u,w\}$ is not an edge in $G$. Second, if $\{u_{1},w_{1}\}$ was previously inserted into~$G'_{1}$, then Invariant~\ref{enu:inv3} says that $\{u,w\}$ is not in any triangle in $G$. Both cases contradict the assumption that $\{u,v,w\}$ is a triangle in $G$. 
    \item If $|K_{1}|\ge3$, then any three vertices in $K_{1}$ corresponds to a triangle in $G$. Thus, we can report a triangle and terminate. 
\end{enumerate}

We continue until we find a triangle or until $G'_{1}$ and $G'_2$ are the complete graph, implying that $G$ has no triangle. Observe that the additional time we spent outside running ${\cal A}$ is proportional to the total recourse of ${\cal A}$, which is at most the total update time $T(2n)$ plus the time to insert edges between $V(G'_{1})$ and $V(G'_{2})$. So, the total running time is $O(T(2n)+n^{2}) = O(T(2n))$ (using the trivial fact that $T(2n) = \Omega(n^2)$ as in the worst case $\mathcal A$ processes $\Omega(n^2)$ updates).
\end{proof}

The above reduction combined with \cref{obs:all-edges-triangle-bmm} immediately implies \Cref{thm:hardness MIS}.

%% file: sections/triangle.tex
\section{Dynamic Triangle Detection} \label{sec:triangle}
In this section we consider the dynamic \emph{triangle detection} problem. In \cref{sec:triangle-decr} we design an $\tilde O(n^\omega)$-time algorithm for \emph{decremental} triangle detection. In \cref{sec:triangle-incr} we show that in contrast \emph{incremental} triangle detection requires total time $n^{3-o(1)}$.

\subsection{Decremental Triangle Detection} \label{sec:triangle-decr}
\thmdectriangle*

The basic idea behind the algorithm is simple: We design an algorithm that proceeds in \emph{stages}. At the beginning of each stage we precompute a set of triangles in $G$---for instance by the following well-known fact:

\begin{lemma}[\cite{AlonGMN92}] \label{lem:triangle-witness}
There is a deterministic $\tilde O(n^\omega)$-time algorithm that, given a graph $G$, computes for each edge $e \in E(G)$ a triangle that contains $e$ or decides that no such triangle exists.
\end{lemma}

We keep reporting arbitrary precomputed triangles as long as the adversary has not destroyed all precomputed triangles. At this point we compute the next set of triangles and proceed to the next stage. The main issue is that there could be edges that are involved in \emph{exceptionally many} precomputed triangles, and so the adversary could potentially destroy all the precomputed triangles by removing only few edges. We avoid this problem by computing a more specific \emph{balanced} set of triangles that avoids this problem.

\subsubsection{Computing a Balanced Set of Triangles}
We start with some definitions to quantify what \emph{exceptionally many} means. Let $T(G)$ denote the set of triangles in a graph~$G$. For an edge~\makebox{$e \in E(G)$}, we write $\tau_G(e) = |\{t \in T(G) : e \subseteq t\}|$ to denote the number of triangles in~$G$ containing $e$. For a triangle~\makebox{$t = \{x, y, z\} \in T(G)$}, we define the \emph{value}~$v_G(t)$ as
\begin{equation*}
    v_G(t) = \frac{1}{\tau_G(\{x, y\})} + \frac{1}{\tau_G(\{y, z\})} + \frac{1}{\tau_G(\{x, z\})},
\end{equation*}
and we define the \emph{value} $v_G(e)$ of an edge $e \in E(G)$ as
\begin{equation*}
    v_G(e) = \frac{1}{3} \cdot \! \sum_{e \subseteq t \in T(G)} v_G(t).
\end{equation*}
Intuitively, these values can be interpreted as follows: Suppose we initially assign one \emph{credit} to each edge $e$ that participates in a triangle. Then $e$ partitions its credit into $\tau_G(e)$ equi-valued parts~$1/\tau_G(e)$ and distributes one such part to each triangle that $e$ is involved in. The total credit received by a triangle $t$ in this way is exactly the value $v_G(t)$. Next, each triangle splits its value into three equi-valued parts and distributes these values to the three incident edges. The total value received by each edge $e$ is exactly its value $v_G(e)$.

We occasionally drop the subscript $G$ whenever it is clear from context. To design our decremental triangle detection algorithm, it turns out that we need to compute a set of triangles in which each edge $e$ appears at most (roughly) $v_G(e)$ times; see the following lemma.

\begin{lemma} \label{lem:triangle-listing-small-value}
There is a randomized $\tilde O(n^\omega)$-time algorithm that, given a graph $G$, computes a set of triangles $T \subseteq T(G)$ such that:
\begin{enumerate}
    \item Each edge $e \in E(G)$ that is in at least one triangle in $G$ appears in at least one triangle~in~$T$.
    \item Each edge $e \in E(G)$ appears in at most $\alpha \cdot v_G(e)$ triangles in $T$, for some $\alpha = (\log n)^{O(1)}$.
\end{enumerate}
\end{lemma}

We recommend skipping the following proof of \cref{lem:triangle-listing-small-value} for now, and to instead imagine that the edges $T$ are selected by the following simple randomized process: For each edge $e$ we sample uniformly and independently a random triangle that contains $e$. The actual proof of \cref{lem:triangle-listing-small-value} is more complicated only to obtain the fast $\tilde O(n^\omega)$ running time.\footnote{Familiar readers might object here, noting that by a standard isolation trick one can easily compute in time~$\tilde O(n^\omega)$, for each edge, a uniformly random triangle. The problem with this standard approach is that these random triangles are \emph{not} independent across different edges, and we thus lack concentration.}

\begin{proof}
By \cref{lem:triangle-witness} we first identify and remove all edges that do not participate in triangles in $G$. Let $E_i = \{e \in E : 2^i \leq \tau(e) < 2^{i+1}\}$ denote the subset of edges that appear in roughly $2^i$ triangles; we can compute the sets $E_i$ in time $O(n^\omega)$ by one matrix multiplication.

The algorithm will proceed in iterations $i \gets 0, \dots, \log n$, where in the $i$-th iteration we will select the triangles for all edges in $E_i$. In the $i$-th iteration we repeat the following steps $O(2^{2i} \log n)$ times: Sample a subset of nodes $V' \subseteq V$ with rate $2^{-i}$. Then apply \cref{lem:triangle-witness} on the induced subgraph~$G[V']$, and for each edge $e \in E_i$ for which we have found a triangle in $G[V']$, we include that triangle into our final list $T$.

\medskip\noindent
\emph{Correctness of Property 1.}
Focus on an arbitrary edge $e = \{x, y\}$. Let $i$ be such that $e \in E_i$, and let~$\{x, y, z_1\}, \dots, \{x, y, z_{2^i}\}$ denote some $2^i$ distinct triangles that $e$ is involved in. We argue that in the $i$-th iteration we find a triangle involving $e$ (with high probability). Indeed, in each repetition the probability that we include $e$ into $V'$ is exactly $2^{-2i}$. Independently, the probability that we include some $z_1, \dots, z_{2^i}$ into $V'$ is $\Omega(1)$. Thus, with probability $\Omega(2^{-2i})$ we include some triangle~$\{x, y, z_j\}$. Since we repeat the process $O(2^{2i} \log n)$ times, with high probability we will indeed find one triangle that involves $e$.

\medskip\noindent
\emph{Correctness of Property 2.}
Fix some edge $e = \{x, y\}$, and focus on some iteration $i$. By Chernoff's bound, with high probability there are only $O(\log n)$ repetitions with $x, y \in V'$; we refer to these as the \emph{relevant} repetitions. Let $\{x, y, z_1\}, \dots, \{x, y, z_{\tau_i}\}$ denote the set of triangles that $e$ is involved in such that~$\{x, y\} \in E_i$ or~$\{x, z_j\} \in E_i$ or $\{y, z_j\} \in E_i$ (these are the only triangles that can potentially be reported in the $i$-th iteration). In each relevant repetition we include each vertex $z_j$ into $V'$ with probability $2^{-i}$, and these events are independent of each other. Thus, in each relevant repetition the expected number of vertices $z_j$ contained in $V'$ is $\tau_i / 2^i$. By Chernoff's bound it follows that with high probability we include at most $O((1 + \tau_i / 2^i) \log n)$ vertices $z_j$ into $V'$ across all relevant repetitions. In this case we clearly report at most $O((1 + \tau_i / 2^i) \log^2 n)$ triangles involving $e$. Summing over all iterations, the total number of reported triangles involving $e$ is indeed
\begin{align*}
    &\sum_{i=1}^{\log n} O\left(\left(\frac{\tau_i}{2^i} + 1\right) \log n\right) \\
    &\qquad= \sum_{e \subseteq \{x, y, z\} \in T(G)} O\left(\left(\frac{1}{\tau(\{x, y\})} + \frac{1}{\tau(\{y, z\})} + \frac{1}{\tau(\{x, z\})}\right) \log n\right) + O(\log^2 n) \\
    &\qquad= \sum_{e \subseteq t \in T(G)} O\left(v(t) \log n\right) + O(\log^2 n) \\
    &\qquad = O(v(e) \log^2 n).
\end{align*}

\medskip\noindent
\emph{Running Time.}
In the $i$-th iteration we invoke \cref{lem:triangle-witness} $\tilde O(2^{2i})$ times on an $O(n / 2^{i})$-node graph, which takes time $\tilde O(2^{2i} \cdot (n / 2^i)^\omega) = \tilde O(n^\omega)$. Thus, also across all $\log n$ iterations the time is $\tilde O(n^\omega)$.
\end{proof}

\subsubsection{The Algorithm}
Let $G_0 \gets G$ denote the initial (complete) graph. The algorithm runs in a polylogarithmic number of \emph{stages}, where we denote by $G_0, G_1, \dots$ the graph at the beginning of the respective stages. In each stage $i$ we follow these steps:
\begin{enumerate}
    \item We call \cref{lem:triangle-listing-small-value} to compute a balanced set of triangles $T_i \subseteq T(G_i)$. We select an arbitrary triangle $t^* \in T_i$ to be the \emph{active} triangle and report $t^*$.
    \item We then process some edge deletions. Upon deletion of an edge $e$, we remove all triangles~\makebox{$t \in T_i$} that contain $e$. If the active triangle $t^*$ contains $e$, we select and report an arbitrary new active triangle from $T_i$. When the set $T_i$ becomes empty (and we thus cannot select a new active triangle), we instead move on to the next step 3.
    \item Call \cref{lem:triangle-witness} to identify all edges that are not involved in triangles, and remove these triangles from the graph. (I.e., we internally mark these edges as removed and if in the future the adversary removes one of these edges the update becomes a no-op.) Let $G_{i+1}$ denote the resulting graph. If $G_{i+1}$ is empty then we can terminate the algorithm and report that no triangle is left. Otherwise we have completed the $i$-th stage, and we proceed to the $(i+1)$-st stage.
\end{enumerate}

\subsubsection{Analysis}
The correctness of this algorithm is immediate: The algorithm keeps reporting triangles until at some point there are no triangles left and we terminate (irrespective of the adversary). The interesting part is to bound the running time by showing that the algorithm indeed terminates after a polylogarithmic number of stages; see the following \cref{lem:triangle-value-total,lem:triangle-value-removed,lem:triangle-stages}.

\begin{lemma} \label{lem:triangle-value-total}
Let $G$ be a graph in which every edge is part of at least one triangle. Then its total value is $\sum_{e \in E(G)} v(e) = \sum_{t \in T(G)} v(t) = |E(G)|$.
\end{lemma}
\begin{proof}
\begin{align*}
    \sum_{e \in E} v(e)
    &= \sum_{e \in E} \,\frac{1}{3} \cdot\! \sum_{e \subseteq t \in T(G)} v(t) \\
    &= \sum_{t \in T(G)} v(t) \\
    &= \sum_{\{x, y, z\} \in T(G)} \left(\frac{1}{\tau(\{x, y\})} + \frac{1}{\tau(\{y, z\})} + \frac{1}{\tau(\{x, z\})}\right) \\
    &= \sum_{\{x, y\} \in E} \frac{\tau(\{x, y\})}{\tau(\{x, y\})} \\
    &= |E|. \qedhere
\end{align*}
\end{proof}

\begin{lemma} \label{lem:triangle-value-removed}
Let $R_i$ be the set of edges removed from $G_i$ to obtain $G_{i+1}$. Then $\sum_{e \in R_i} v_{G_i}(e) \geq \frac{|E(G_i)|}{\alpha}$.
\end{lemma}
\begin{proof}
This is by design: Each edge $e$ appears in at most $\alpha \cdot v(e)$ of the selected triangles $T_i$. Thus, the adversary is forced to remove edges of total value at least~\smash{$\frac{|E(G_i)|}{\alpha}$} to destroy the at least $|E(G_i)|$ triangles in $T_i$.
\end{proof}

\begin{lemma} \label{lem:triangle-stages}
The algorithm terminates after at most $O(\alpha^2 \log^2 n)$ stages.
\end{lemma}
\begin{proof}
Focus on any stage $i$. We show that transitioning to the $(i+1)$-st stage we make progress in at least one of two progress measures: The number of edges (specifically, $|E(G_{i+1})| \leq (1 - \Omega(\frac{1}{\alpha})) |E(G_i)|$), or the number of triangles (specifically, for at least an $\Omega(\frac{1}{\alpha})$-fraction of edges $e \in E(G_{i+1})$ we have that $\tau_{G_{i+1}}(e) \leq (1 - \Omega(\frac{1}{\alpha})) \tau_{G_i}(e)$). The former event can clearly happen at most $O(\alpha \log n)$ times before the graph becomes empty. The latter event similarly can happen at most $O(\alpha^2 \log^2 n)$ times before the graph contains no more triangles (and is thus empty by construction). 

To see this, suppose that the number of edges does not decrease much, $|E(G_{i+1})| > (1 - \frac{1}{4\alpha}) |E(G_i)|$, say. Observe that each triangle $t$ that appears in $G_i$ but not in $G_{i+1}$ contains at least one edge (and at most three edges) from $R_i$. Therefore, and applying the previous \cref{lem:triangle-value-total,lem:triangle-value-removed}, it follows that:
\begin{align*}
    \qquad\qquad\sum_{t \in T(G_{i+1})} v_{G_i}(t)
    &= \sum_{t \in T(G_i)} v_{G_i}(t) - \sum_{t \in T(G_i) \setminus T(G_{i+1})} v_{G_i}(t) \\
    &\leq \sum_{t \in T(G_i)} v_{G_i}(t) - \frac{1}{3} \cdot \sum_{e \in R_i} \sum_{e \subseteq t \in T(G_i)} v_{G_i}(t) \\
    &= \sum_{t \in T(G_i)} v_{G_i}(t) - \sum_{e \in R_i} v_{G_i}(e) \\
    &\leq \left(1 - \frac{1}{\alpha}\right)|E(G_i)| \\
    &\leq \left(1 - \frac{1}{2\alpha}\right)|E(G_{i+1})|.
\intertext{On the other hand, we have that:}
    \sum_{t \in T(G_{i+1})} v_{G_i}(t)
    &= \sum_{t = \{x, y, z\} \in T(G_{i+1})} \left(\frac{1}{\tau_{G_i}(\{x, y\})} + \frac{1}{\tau_{G_i}(\{y, z\})} + \frac{1}{\tau_{G_i}(\{x, z\})}\right) \\
    &= \sum_{e \in E(G_{i+1})} \frac{\tau_{G_{i+1}}(e)}{\tau_{G_i}(e)}.
\end{align*}
Putting both calculations together, it follows that for a random edge $e \in E(G_{i+1})$ we expect that the ratio~$\tau_{G_{i+1}}(e) / \tau_{G_i}(e)$ is at most $1 - \frac{1}{2\alpha}$. From Markov's inequality it follows that, for a random edge $e \in E(G_{i+1})$, with probability at least~\smash{$1 - \frac{1}{1 + 1/4\alpha} = \Omega(\frac{1}{\alpha})$} we have
\begin{equation*}
    \frac{\tau_{G_{i+1}}(e)}{\tau_{G_i}(e)} \leq \left(1 - \frac{1}{2\alpha}\right) \left(1 + \frac{1}{4\alpha}\right) \leq 1 - \frac{1}{4\alpha}.
\end{equation*}
In other words, for an $\Omega(\frac{1}{\alpha})$-fraction of edges~\makebox{$e \in E(G_{i+1})$} the number of triangles that $e$ is involved in indeed increases by $\tau_{G_{i+1}}(e) \leq (1 - \Omega(\frac{1}{\alpha})) \tau_{G_i}(e)$. This completes the proof.
\end{proof}

\begin{proof}[Proof of \cref{thm:dec triangle}]
In light of the previous \cref{lem:triangle-stages} it easily follows that our decremental triangle detection algorithm runs in time $\tilde O(n^\omega)$: Each stage takes time $\tilde O(n^\omega)$ to run \cref{lem:triangle-listing-small-value,lem:triangle-witness}. To efficiently process the edge deletions, we additionally store for each edge $e$ a list of pointers to the triangles in $T_i$ that $e$ is involved in. This list can easily be prepared in time $O(|T_i|)$ at the beginning of each stage. When an edge $e$ is deleted we traverse this list and remove all affected triangles from $T_i$. It follows that the total time to deal with the edge deletions is $O(n^2 + |T_i|) = \tilde O(n^2)$, where the upper bound $|T_i| = \tilde O(n^2)$ follows from \cref{lem:triangle-value-total}. As the total number of stages is $O(\alpha^2 \log^2 n) = (\log n)^{O(1)}$, this completes the proof.
\end{proof}

\subsection{Incremental Triangle Detection} \label{sec:triangle-incr}
We complement our \emph{decremental} triangle detection algorithm with a conditional lower bound for \emph{incremental} triangle detection ruling out polynomial improvements over the brute-force $O(n^3)$ algorithm.

\thminctriangle*

We note that for the incremental triangle detection problem it does not matter whether the adversary is oblivious or adaptive since the output is fixed (the output is NO until the first triangle is introduced, and YES afterwards). To prove \cref{thm:inc triangle} we rely on the following equivalent variant of the OMv hypothesis~\cite{HenzingerKNS15}:

\begin{hypothesis}[OuMv]
Every algorithm that preprocesses a given bipartite graph $G = (X, Y, E)$ and then answers $n$ online queries of the form \emph{``given $X' \subseteq X, Y' \subseteq Y$, is there an edge in the induced subgraph $G[X', Y']$?''} takes total time $n^{3-o(1)}$.
\end{hypothesis}

\begin{proof}[Proof of \cref{thm:inc triangle}]
Let $G = (X, Y, E)$ denote a given OuMv instance. In the preprocessing phase we partition $X$ and $Y$ arbitrarily into $g$ groups~\makebox{$X = X_1 \sqcup \dots \sqcup X_g$} and~\makebox{$Y = Y_1 \sqcup \dots \sqcup Y_g$} of size~$O(n/g)$ each (for some parameter $g$ to be determined later). For each pair $(i, j) \in [g]^2$ we construct an incremental triangle detection instance~$I_{i, j}$ as follows: $I_{i, j}$ is a tripartite graph with vertex parts $X_i$, $Y_j$ and~$Z_{i, j}$, where~$Z_{i, j}$ is a fresh set of $\Theta(n/g)$ nodes. Initially we only include the edges in $G[X_i, Y_j]$. We maintain the invariants that each instance $I_{i, j}$ does \emph{not} contain a triangle, and contains at least one \emph{isolated} node in $Z$.

Next, we describe how to implement a query $X' \subseteq X, Y' \subseteq Y$. Enumerate the pairs $(i, j) \in [g]^2$ in an arbitrary order. For each pair, we take an isolated node $z \in Z_{i, j}$, and insert edges from $z$ to all nodes in $X' \cap X_i$ and to all nodes in $Y' \cap Y_j$ in $I_{i, j}$. Note that we introduce a triangle to the instance~$I_{i, j}$ if and only if there is an edge in the induced graph $G[X' \cap X_i, Y' \cap Y_j]$. Thus, if the triangle detection algorithm reports YES, we can stop and report YES for the query $(X', Y')$. Otherwise we continue with the next pair $(i, j)$.

To maintain our invariants, after each query we \emph{reset} all instances $I_{i, j}$ that contain a triangle or that do not contain an isolated node in $Z_{i, j}$ (i.e., we set up a fresh instance $I_{i, j}$ as in the preprocessing phase).

The correctness is clear. To bound the running time, we first analyze how often we reset an instance $I_{i, j}$. On the one hand, each query introduces a triangle to at most one instance $I_{i, j}$ leading to at most $n$ resets. On the other hand, each query decreases the number of isolated nodes in $Z_{i, j}$ by at most 1, hence we reset $G_{i, j}$ only $n / \Theta(n / g) = O(g)$ times. Across all pairs of groups, this leads to $O(g^3)$ resets. Therefore, the total number of resets is $O(n + g^3)$. Assuming that incremental triangle detection can be solved in time $O(n^{3-\epsilon})$ for some $\epsilon > 0$, it thus follows that we can process the OuMv queries in total time $O((n + g^3) \cdot (n/g)^{3-\epsilon})$. Setting $g := \lceil n^{1/3}\rceil$ this becomes $O(n^{3-2\epsilon/3})$, contradicting the OuMv hypothesis and thus the OMv hypothesis~\cite{HenzingerKNS15}.
\end{proof}

%% file: sections/clique-algo.tex
\section{Decremental Maximal Clique} \label{sec:clique}
In this section we prove \cref{thm:alg clique}, i.e., we show that one can maintain maximal cliques in all connected components in a decremental graph against an \emph{oblivious adversary}. The proof is in two steps: We first show that one can efficiently maintain a maximal clique in \emph{one} connected component by a simple reduction to dynamic MIS (\cref{sec:clique:sec:one-comp}), and then generalize the algorithm to apply to \emph{all} connected components (\cref{sec:clique:sec:all-comps}).

\subsection{Decremental Maximal Clique in One Connected Component} \label{sec:clique:sec:one-comp}
Let $G = (V, E)$ be a graph. Recall that a \emph{maximal clique (MC)} is a clique $K \subseteq V$ such that for every node $y \in V \setminus K$, there is some node $x \in K$ such that $(y, x) \notin E$. Recall also that a maximal independent set of the complement of $G$ is a maximal clique of $G$ (\cref{obs:mis-complement}).

As $G$ undergoes edge deletions, the complement graph $G^c$ undergoes edge insertions. This would intuitively suggest that a $\tilde{O}(1)$ update time algorithm for incremental maximal independent set can be used in a black-box way to obtain a decremental maximal clique algorithm with the same update time. Unfortunately, a naive reduction to incremental MIS would result in an algorithm with total update time proportional to the size of the complement graph $G^c$, which can be significantly denser than $G$. We show how to avoid this blow-up by a simple trick.

\begin{theorem}
\label{lem:clique:reduction}
Assume that there is a dynamic algorithm $\mathcal A$ maintaining a maximal independent set in an initially empty $n$-vertex graph that undergoes a series of $m$ edge insertions and deletions in $T(n,m)$ total update time. Then there is a dynamic algorithm $\mathcal B$ maintaining a maximal clique in an $n$-vertex $m$-edge graph that undergoes a series of edge deletions in $O(T(n, O(m)))$ total update time. Moreover, if $\mathcal A$ is robust against adaptive adversaries then so is $\mathcal B$.
\end{theorem}
\begin{proof}
We start with a description of $\mathcal B$. Let $G = (V, E)$ denote the given graph. At initialization, we select a lowest-degree vertex $v$ as the \emph{pivot} vertex. We construct the graph $G' = G[N_G(v)]^c$; i.e., the complement of the subgraph induced by the neighborhood of $v$ (excluding $v$ itself). We run the dynamic MIS algorithm $\mathcal A$ on $G'$ to maintain some MIS $S$. At all times we will output the maximal clique $\{v\} \cup (S \cap N_G(v))$.

Whenever an edge $e$ is deleted from $G$, we distinguish three cases: If $e$ involves the designated node $v$, say $e = \{v, w\}$, then we remove all edges incident to $w$ in $G'$.  Let us call such a vertex~$w$ \emph{eliminated}. Otherwise, if $e$ involves only nodes that are present in $G'$, then insert $e$ into $G'$. Otherwise, the edge deletion does not affect $G'$.

Observe that $S$ at all times will consist of the union of an MIS of $G[N_G(v)]^c$ and the eliminated vertices. Hence, the algorithm correctly outputs an maximal clique of $G$. Furthermore, the updates to $G'$ do not depend on the output of $\mathcal A$.

We finally analyze the total update time. Finding $v$ will take $O(m)$ time at initialization. Any potential edge between vertices of $G'$ will be inserted into and removed from $G'$ at most once, hence the there are at most $O(\deg_G(v)^2)$ updates in total. Consequently, the total update time of $\mathcal A$ over the $m$ deletions is $T(\deg_G(v), O(\deg_G(v)^2))$. Finally, recall that $v$ is the smallest-degree vertex of~$G$ (at initialization) and thus $\deg_G(v)^2 \leq n \cdot \deg_G(v) \leq 2 m$. Hence, the total update time of the algorithm is $O(T(n, O(m)))$.
\end{proof}

Combining Lemma~\ref{lem:clique:reduction} with the known dynamic MIS algorithms in total update time $\tilde{O}(n + m)$~\cite{ChechikZ19,BehnezhadDHSS19} yields the following theorem:

\begin{theorem} \label{th:main:clique}
There is a dynamic algorithm maintaining a maximal clique of a decremental graph in $\tilde{O}(n+m)$ expected total update time against an oblivious adversary.
\end{theorem}

In addition, it is known how to maintain a maximal independent set of a fully dynamic graph on~$n$ vertices undergoing $m$ edge updates in $O(nm)$ total update time against an adaptive adversary (for an example see~\cite{AssadiOSS18}). If we substitute such an adaptive algorithm into the reduction we receive a total update time bound of $O(m + \min_{v \in V} \deg_G(v)^3)$. Amortized over the $m  = \Omega(\min_{v \in V} \deg_G(v)^2)$ edge deletions we get $O(\min_{v \in V}\deg_G(v)) = O(\Delta)$ amortized update time.

\begin{theorem}
There is a dynamic algorithm maintaining a maximal clique of a decremental graph~$G$ with maximal degree $\Delta$ in $O(\Delta)$ amortized update time against an adaptive adversary.
\end{theorem}

Finally, we remark that in the proof of \cref{th:main:clique} it was not strictly necessary to take a pivot with minimal degree. In fact, we have the freedom to pick any pivot node, then the total update time is $O(m + T(\deg(v), O(\deg(v)^2)))$. In particular, it would have sufficed to take any node with degree $\deg(v) \leq O(\bar d)$ where $\bar d$ is the \emph{average degree}, as in this case $O(\deg(v)^2) = O(\bar d^2) = O(n \bar d) = O(m)$. We will rely on this observation in the next subsection.

\subsection{Decremental Maximal Clique in Every Connected Component} \label{sec:clique:sec:all-comps}
If the input graph $G = (V, E)$ is not connected, then it might admit some trivial maximal cliques. For example, suppose that $G$ consists of an isolated node $v$ (with degree zero), and the set of remaining nodes $V \setminus \{v\}$ forms a clique. Then, the singleton set $\{v\}$ forms a maximal clique in $G$. We need to be especially mindful of this scenario when we are dealing with a decremental graph, since after sufficiently many edge deletions the graph is guaranteed to become disconnected. With this backdrop, we now strengthen the result derived in  \Cref{th:main:clique} as follows.

We define a set $K \subseteq V$ to be a \emph{maximal clique in every connected component (MCCC)} of an input graph $G = (V, E)$ if for every connected component $C$ of $G$, the subset $K \cap C$ is a maximal clique in $G[C]$. We devote the rest of this section towards proving the following theorem, which generalizes \Cref{th:main:clique} to MCCC, in addition to obtaining a high probability bound on the total update time.

\begin{theorem}\label{th:main:clique:component}
There is a randomized algorithm for maintaining a maximal clique in every connected component in a decremental $n$-vertex graph which undergoes a sequence of edge deletions generated by an \emph{oblivious} adversary. The algorithm has $\tilde{O}(n+m)$ total update time with high probability, where~$m$ denotes the initial number of edges.
\end{theorem}

We first consider the slightly simpler problem of maintaining a maximal clique in the unique largest connected component of size larger than $\frac{n}{2}$ (if it exists):

\begin{lemma} \label{lm:intermediate}
There is a randomized algorithm $\mathcal B^\star$ which maintains in a decremental graph a connected component $C^\star$ of size \smash{$|C^\star| > \frac{n}{2}$} along with a maximal clique $K^\star \subseteq C^\star$ in $G[C^\star]$, as long as such a component $C^\star$ exists. It runs in $\tilde{O}(n+m)$ expected total update time.
\end{lemma}
\begin{proof}
Using the classic dynamic algorithm by Holm, Lichtenberg and Thorup~\cite{HolmLT01} we can throughout maintain the connected components in $G$. In particular, we can maintain the largest connected component~$C^\star$ as long as it has size $|C^\star| > \frac{n}{2}$. It remains to maintain a maximal clique~\makebox{$K^\star \subseteq C^\star$}.

Let $\mathcal B$ be the decremental algorithm from \cref{th:main:clique} (where we plug in any MIS algorithm with linear total update time~\cite{ChechikZ19,BehnezhadDHSS19}). Let $\gamma = \Theta(\log n)$ and let $\mathcal B_1, \dots, \mathcal B_\gamma$ denote independent copies of $\mathcal B$ where in each copy $\mathcal B_i$ we pick a \emph{uniformly random} pivot node $v_i$. To process the edge deletions, we will run all copies $\mathcal B_1, \dots, \mathcal B_\gamma$ in parallel. Let $K_i$ denote the maximal clique maintained by $\mathcal B_i$. At any point during the execution we call~$i$ \emph{valid} if $v_i \in C^\star$ (note that this can be tested in constant time). After each update we take the smallest valid index $i$, and report $K^\star := K_i$.

It is easy to verify that the algorithm is correct, provided that there always is some index $i$ that is valid. Indeed, if $i$ is valid then $K_i \subseteq C^\star$ is a maximal clique (as $\mathcal B_i$ maintains a maximal clique in the connected component of $v_i$). To see that there always is a valid index, recall that each pivot $v_i$ is chosen uniformly at random, and thus part of $C^\star$ with probability at least $\frac{1}{2}$. (Here we are crucially use that the adversary is oblivious.) It follows that there is a valid index with high probability~\makebox{$1 - (\frac{1}{2})^{\gamma} = n^{-\Omega(1)}$}.

We finally analyze the expected running time. First, note that each copy $\mathcal B_i$ runs in expected time $\tilde O(n + m)$. To see this, recall that the pivot vertex $v_i$ is chosen uniformly at random and thus its expected degree is $\Ex[\deg(v)] = \bar d$. It follows that the expected running time is $\tilde O(n + \deg(v)^2) = \tilde O(n + n \deg(v)) = \tilde O(n + n \bar d) = \tilde O(n + m)$. In addition, to report the set $K^\star$ the algorithm spends time proportional to the recourse of the $\mathcal B_i$'s. The only exception is when after some update the active index $i$ changes. However, note that this can happen at most $\gamma = O(\log n)$ times (as clearly all invalid indices $i$ will stay invalid for all future updates) and so the overhead to deal with this exceptional case is $O(n \log n)$.
\end{proof}

\begin{proof}[Proof of \cref{th:main:clique:component}]
We design a recursive dynamic algorithm $\mathcal C$ to maintain maximal cliques in all connected components. Throughout we run~\cite{HolmLT01} to dynamically maintain the connected components in the decremental graph $G$. Let $C_1, \dots, C_\ell$ denote the connected components of size at most $\frac{n}{2}$ at initialization, and let $C^\star$ denote the unique component of size more than $\frac{n}{2}$ (if it exists). We run $\mathcal C$ recursively on $G[C_1], \dots, G[C_\ell]$. And we run the algorithm $\mathcal B^\star$ from \Cref{lm:intermediate} on $G[C^\star]$.

We process the edge deletions as follows. Each edge deletion outside $C^\star$ is simply passed to the respective recursive call. Each edge deletion inside $C^\star$ is passed to $\mathcal B^\star$. We report the union of all the sets reported by the recursive calls and by $\mathcal B^\star$. If as the result of some edge deletion the connected component $C^\star$ becomes disconnected into two connected components~$C', C''$ with $|C'| \geq |C''|$, say, we distinguish two cases. If $|C'| \leq \frac{n}{2}$, then we stop the execution of $\mathcal B^*$ and we recursively call $\mathcal C$ on~$G[C']$ and~$G[C'']$. If $|C'| > \frac{n}{2}$ then we keep executing $\mathcal B^\star$ on $C^\star \gets C'$, and we call $\mathcal C$ recursively only on~$G[C'']$.

It is easy to argue inductively that the algorithm throughout reports a MCCC. It remains to bound the total update time of $\mathcal C$. We start with a bound on the expected update time, and later comment how to obtain a high-probability bound. Let $T_{\mathcal C}(n, m)$ denote the expected running time of~$\mathcal C$, and let $R_{\mathcal C}(n, m)$ denote the expected recourse of $\mathcal C$, and define $T_{\mathcal B^\star}(n, m)$ and $R_{\mathcal B^\star}(n, m)$ similarly. Let $n_1, \dots, n_k$ and $m_1, \dots, m_k$ denote the number of nodes and edges in the recursive calls; clearly we have $\sum_i n_i \leq n$ and $\sum_i m_i \leq m$. We have the following recurrences:
\begin{align*}
    R_{\mathcal C}(n, m) &\leq R_{\mathcal B^\star}(n, m) + \sum_{i=1}^k R_{\mathcal C}(n_i, m_i), \\
    T_{\mathcal C}(n, m) &\leq O(T_{\mathcal B^\star}(n, m)) + \sum_{i=1}^k T_{\mathcal C}(n_i, m_i) + \tilde O\left(\sum_{i=1}^k R_{\mathcal C}(n_i, m_i)\right).
\end{align*}
Recalling that $R_{\mathcal B^\star}(n, m) \leq T_{\mathcal B^\star}(n, m) = \tilde O(n + m)$, the first recurrence solves to $R_{\mathcal C}(n, m) = \tilde O(n + m)$ and therefore the second recurrence solves to $T_{\mathcal C}(n, m) = \tilde O(n + m)$ as well.

Finally, we can turn $\mathcal C$ in a black-box manner into an algorithm that runs in total update time $\tilde O(n + m)$ \emph{with high probability}. The idea is standard: We run $\Theta(\log n)$ independent copies of $\mathcal C$, and interrupt each copy that takes total time more than $2 T_{\mathcal C}(n, m)$. By Markov's inequality each copy runs in time at most $2 T_{\mathcal C}(n, m)$ with probability at least $\frac{1}{2}$, and so with high probability there is indeed at least one copy which is not interrupted. The total time is \smash{$\tilde O(T_{\mathcal C}(n, m)) = \tilde O(n + m)$}.
\end{proof}

%% file: sections/open.tex
\section{Conclusion and Open Problems} \label{sec:open}
\Cref{thm:hardness MIS,thm:hardness clique} give the first update-time separation between dynamic algorithms against oblivious adversaries and those against adaptive adversaries in \emph{natural} dynamic problems with explicit input based on popular fine-grained complexity hypotheses. We hope these results will pave the way for further separation of natural dynamic problems. Below, we suggest exciting open problems and then survey related interesting separations in this area and their relation to our results.

\paragraph{More Natural Separations.}
Can we show more oblivious-vs-adaptive separations for natural dynamic problems? The next natural targets are other symmetric breaking problems besides MIS. It would be very interesting to see separations for dynamic maximal matching or $(\Delta+1)$-vertex coloring.

\paragraph{Improved Robust Dynamic MIS.}
In addition to our main open problem of achieving an exponential separation for dynamic MIS when $\omega=2$, our lower bounds motivate several interesting algorithmic avenues.

Although the update time is tight at $n^{1\pm o(1)}$ for \emph{combinatorial} robust algorithms, can we achieve~$n^{0.99}$ update time using fast matrix multiplication? Even within combinatorial approaches, how much can we leverage graph sparsity? A deterministic fully dynamic algorithm~\cite{AssadiOSS18} has $O(\min\{\Delta,m^{2/3}\})$ update time, where $\Delta$ and $m$ are the maximum degree and the number of edges, respectively. This bound improves to $O(\min\{\Delta,m^{1/2}\})$ in the incremental setting~\cite{GuptaK21}. Another deterministic fully dynamic algorithm by~\cite{OnakSSW20} has $\tilde{O}(\alpha^{2})$ update time, where $\alpha\le\min\{\Delta,m^{1/2}\}$ is the arboricity. Is there a fully dynamic algorithm with $\tilde{O}(\alpha)$ update time? This would subsume all these prior results.

\subsection*{Related Separations}
Finally, we survey known separation results related to ours.

\paragraph{Oblivious vs Adaptive: Recourse.}
It is known that non-robust dynamic algorithms can outperform robust ones in \emph{recourse}, i.e., the number of updates to the algorithm's output. For example, there are non-robust algorithms for decremental spanners~\cite{BaswanaKS12} and incremental $k$-median~\cite{FichtenbergerLN21} whose recourse is \emph{sublinear in the number of updates}, while every robust algorithm requires linear recourse. However, this is not useful for update-time separation, as every algorithm must spend \emph{time} at least linear in the number of updates.\footnote{Recourse trivially lower bounds the update time for every dynamic problem on \emph{explicitly maintaining} objects such as spanners or maximal independent sets. Another type of dynamic problems is to support \emph{answering queries} such as ``is a given edge in a maximal matching?''. In these problems, recourse is not even defined, but, informally, the ``internal recourse'' is possibly larger than the update time. For example, data structures may maintain several objects and switch the pointer between these objects to answer queries quickly (see e.g.~\cite{BernsteinFH21}).}

\paragraph{Oblivious vs Deterministic.}
With space restrictions, there is a separation between non-robust and deterministic dynamic algorithms. In the list labeling problem for $n$ items and $m=O(n)$ space, every deterministic algorithm requires $\Omega(\log^{2}n)$ amortized time~\cite{BulanekKS15}, whereas there exists a randomized algorithm against an oblivious adversary with $O(\log n(\log\log n)^{3})$ amortized time~\cite{BenderCFKKW22,BenderCFKKKS24}. However, with larger space $m=n^{1+\Theta(1)}$, the optimal time of both deterministic and randomized algorithms coincides at $\Theta(\log n)$~\cite{BulanekKS13}.

\paragraph{Randomized vs Deterministic: Static.}
The sublinear-time algorithm community has shown that various problems admit randomized static algorithms with sublinear time when the input is prepared as an oracle. For most such problems, deterministic sublinear-time algorithms cannot exist due to the standard argument that the adversary has complete freedom to choose the unrevealed parts of the input to deterministic algorithms (see, e.g.,~\cite{Goldreich17}).

In contrast, above the sublinear-time regime, there are almost no results separating deterministic algorithms from randomized ones. The exception is the univariate polynomial identity testing problem~\cite{Williams16}, which admits a near-linear time randomized algorithm but requires quadratic time deterministically assuming NSETH~\cite{Williams16}. The same holds for the tropical tensor problem~\cite{AbboudR18}, assuming that we cannot prove certain circuit lower bounds.

\paragraph{Incremental vs Decremental.}
To our knowledge, the only prior incremental-vs-decremental separation similar to ours (\Cref{thm:inc triangle,thm:dec triangle}) was in the context of \emph{sensitivity oracles}; unlike classic dynamic algorithms, these data structures can handle only a \emph{single} batch of updates. Given a graph with $n$ vertices and $m$ edges, there exists a connectivity oracle with $m^{1+o(1)}$ preprocessing time and space that handles $d$ vertex deletions in $\tilde{O}(d^{2})$ update time and can then answer any pairwise-connectivity query in $O(d)$ time~\cite{DuanP20,LongS22,LongW24}. However, for any growing function $f(\cdot)$, every connectivity oracle for handling $d$ vertex insertions with $f(d)n^{o(1)}$ update and query time must take $\Omega(n^{\omega})$ preprocessing time and $\Omega(n^{2})$ space, assuming fine-grained complexity hypotheses~\cite{LongW24}.\footnote{Lemmas 7.3 and 7.10 of~\cite{LongW24} were stated for the fully dynamic setting, but they also hold in the incremental setting.}

\paragraph{Amortized vs Worst-Case.}
There is an amortized-vs-worst-case separation for the fully dynamic MIS problem: it admits algorithms with $\polylog(n)$ \emph{expected worst-case} update time~\cite{ChechikZ19,BehnezhadDHSS19}, which implies $\polylog(n)$ amortized update time. However, every incremental MIS algorithm requires $\Omega(n)$ worst-case recourse, and hence, $\Omega(n)$ worst-case update time~\cite{AssadiOSS18}.